\RequirePackage{fix-cm}
\documentclass[a4paper,11pt]{article}

\usepackage{authblk}
\usepackage[utf8]{inputenc}
\usepackage[T1]{fontenc}
\usepackage{pslatex}
\usepackage{amsfonts,amsmath,amssymb,amsthm}
\usepackage{algorithm}
\usepackage{paralist}
\usepackage{booktabs}
\usepackage{url}
\usepackage{hyperref}
\usepackage{graphicx}
\usepackage{microtype}
\usepackage{url}
\usepackage{hyperref}
\usepackage{makecell}
\usepackage{subcaption}
\usepackage{tikz}
\usepackage{multirow}
\usetikzlibrary{positioning}
\usetikzlibrary{decorations.pathreplacing}
\newcommand{\N}{\mathbb{N}}

\newcommand{\F}{\mathbb{F}}

\newtheorem{definition}{Definition}
\newtheorem{lemma}{Lemma}
\newtheorem{theorem}{Theorem}
\newtheorem{remark}{Remark}

\newtheorem{example}{Example}

\tikzset{
    mybrace/.style={decorate,decoration={brace,aspect=#1}}
}

\providecommand{\keywords}[1]{\textbf{\textit{Keywords }} #1}

\begin{document}

\title{How to reconstruct (anonymously) a secret cellular automaton}

\author[1]{Luca Mariot}
\author[1]{Federico Mazzone}
\author[2]{Luca Manzoni}
\author[3]{Alberto Leporati}
	
\affil[1]{{\small Semantics, Cybersecurity and Services Group, University of Twente, Drienerlolaan 5, 7511GG Enschede, The Netherlands} 
	
	{\small \texttt{\{l.mariot, f.mazzone\}@utwente.nl}}}

\affil[2]{{\small Department of Mathematics, Informatics and Geosciences, University of Trieste, Via Valerio 12/1, 34127 Trieste, Italy} 
	
	{\small \texttt{lmanzoni@units.it}}}

\affil[3]{{\small Department of Informatics, Systems and Communication, University of Milano-Bicocca, Viale Sarca 336/14, 20126 Milano, Italy} 
	
	{\small \texttt{alberto.leporati@unimib.it}}}

\maketitle

\begin{abstract}
We consider threshold secret sharing schemes based on cellular automata (CA) that allows for anonymous reconstruction, meaning that the secret can be recovered only as a function of the shares, without knowing the participants' identities. To this end, we revisit the basic characterization of $(2,n)$ threshold schemes based on CA in terms of Mutually Orthogonal Latin Squares (MOLS), and redefine the secret space as the MOLS family itself, showing that the new resulting scheme enables anonymous reconstruction of secret CA rules. Finally, we discuss the trade-off between the number of secret CA that can be shared and the computational complexity of the recovery phase.
\end{abstract}


\section{Introduction}
\label{sec:intro}
Secret Sharing Schemes (SSS) are a fundamental cryptographic primitive underlying several key management and secure multiparty computation protocols~\cite{goldreich87,sahai05}. Informally, a SSS enables a \emph{dealer} to share a secret value $S$ among a set of $n$ \emph{participants} (or \emph{players}), by distributing to them so-called \emph{shares} of $S$. The sharing is done in such a way that only certain \emph{authorized subsets} of players can uniquely determine the secret by pooling together their shares.

Historically, SSS have been introduced by Shamir~\cite{shamir79} and Blakley~\cite{blakley79} in the setting of $(t,n)$ \emph{threshold access structures}, meaning that the authorized subsets are all those of cardinality at least $t\le n$. When $t=2$ (i.e., only two players are required to reconstruct the secret), threshold SSS have a nice characterization in terms of combinatorial designs. In fact, one can show that $(2, n)$ perfect schemes are equivalent to sets of $n$ \emph{Mutually Orthogonal Latin Squares} (MOLS)~\cite{stinson-cd}.

The idea of adopting Cellular Automata (CA) to design threshold secret sharing schemes dates back at least two decades. However, early works in this research thread~\cite{delrey05,mariot14} came up with threshold SSS with an additional \emph{adjacency} constraint on the access structure, meaning that the shares must be contiguous CA configurations to recover the secret. Moreover, as shown by~\cite{herranz25}, such schemes cannot be perfect, unless in very trivial cases (i.e. when $t=1$ or $t=n$). On the other hand, Mariot et al.~\cite{mariot20} studied how to construct families of MOLS using Linear Bipermutive CA (LBCA), thereby yielding a perfect SSS with a full $(2,n)$ threshold access structure.

In this paper, we investigate the design of anonymous secret sharing schemes based on CA. To this end, we start from the construction of MOLS introduced in~\cite{mariot20}, observing that the $(2,n)$-threshold scheme of~\cite{mariot18} derived from this construction requires the players to disclose both their shares and associated CA rules. Then, we observe that by redefining the set of possible secrets from the row space of the Latin squares to the family of MOLS itself allows one to partition the secrets into \emph{parallel classes}, enabling anonymous reconstruction. From a practical standpoint, this entails that the players are given shares of a \emph{secret CA local rule}, and the task of any two players is to reconstruct the rule without revealing to one another their respective shares. We then describe the resulting anonymous $(2,n)$-scheme protocol, which include both preimage computation and forward evolution of the underlying CA. Finally, we remark that the efficiency of this scheme is characterized by a trade-off between the number of secrets that can be shared and the computational complexity of the recovery phase. 

The remainder of this paper is organized as follows. Section~\ref{sec:bg} covers preliminary definitions about cellular automata, (anonymous) secret sharing schemes and orthogonal Latin squares used in the paper. Section~\ref{sec:scheme} recalls the basic $(2,n)$ threshold scheme introduced in~\cite{mariot18}, showing why it cannot be anonymous and describes the new scheme, proving its correctness and security properties. Finally, section~\ref{sec:outro} discusses the trade-off between the number of secrets and the computational complexity of the recovery phase, and describes some interesting open problems for further research on the subject.

\section{Background}
\label{sec:bg}
In this section, we cover the basic definitions and results needed in the later sections of the paper. We start by introducing cellular automata as algebraic systems, and continue with a brief description of secret sharing schemes and their connection to orthogonal Latin squares.

\subsection{Cellular Automata}
\label{subsec:ca}
Cellular Automata (CA) are usually described as a particular type of discrete dynamical systems, characterized by a regular lattice of cells whose global state is updated by applying the same local rule in a shift-invariant manner to each cell. Typical questions investigated in this setting revolve around the long-term dynamics induced by the repeated application of the \emph{CA global rule} over multiple time steps, for example the characterization of its fixed points and limit cycles. On the other hand, in this work we consider CA under the perspective of \emph{algebraic systems}, since we are interested only in a single application of the CA global rule~\cite{manzoni26}. In particular, we formally define a \emph{no boundary CA} (NBCA) as follows:

\begin{definition}
	\label{def:ca}
	Let $\Sigma$ be a finite set of size $q \in \N$, and let $d, n \in \N$ such that $d\le n$. Further, let $f: \Sigma^d \to \Sigma$ be a $d$-variable mapping over $\Sigma$. Then, a \emph{no-boundary CA} over the alphabet $\Sigma$ with $n$ input cells and \emph{local rule} $f$ of \emph{diameter} $d$ is a mapping $F: \Sigma^{n} \rightarrow \Sigma^{n-d+1}$ defined for all $x \in \Sigma^n$ as:
	\begin{equation}
		\label{eq:nbca}
		F(x_1, \ldots, x_n) = (f(x_1, \ldots, x_d), f(x_2, \ldots, x_{d+1}),
		\ldots, f(x_{n-d+1}, \ldots, x_n)) \enspace .
	\end{equation}
\end{definition}
In other words, a NBCA is a particular type of vectorial function where each output coordinate $i \in \{1,\cdots, n\}$ is defined as the application of the CA local rule $f$ to the neighborhood composed of the $i$-th input cell and the $d-1$ cells to its right. Notice that the global rule of a NBCA cannot be iterated for an indefinite number of steps, since the size of the cellular array ``shrinks'' by $d-1$ cells at each iteration. Other models, such as \emph{periodic boundary CA}~\cite{mariot19}, retain the whole cellular array from one iteration to the next, turning effectively a CA into a discrete dynamical system. However, as remarked earlier, we will be interested only in NBCA seen as algebraic systems, with a single application of the CA global rule $F$. 

The most natural way to represent the local rule $f: \Sigma^d \to \Sigma$ is by means of a lookup table of $q^d$ rows, which defines the next state $f(x)$ of a cell for each configuration $x \in \Sigma^d$ of its neighborhood. When $\Sigma = \F_2 = \{0, 1\}$, the lookup table is a Boolean function of $d$ variables. In this case, the decimal encoding of the truth table's output column is also called the \emph{Wolfram code} of the rule~\cite{wolfram83}.

Another useful way to encode the local rule $f$ of a CA is the \emph{de Bruijn} graph $G_f = (V,E)$, where $V = \Sigma^{d-1}$ and given $u,v \in V$ the pair $(u,v)$ is connected by a directed edge if and only if they \emph{overlap} respectively on the rightmost and leftmost $d-2$ coordinates. A local rule $f: \Sigma^d \to \Sigma$ can then be represented as a labeling function $l_f: E \to \F_2$ on the edges of $G_f$, where for each $(u,v) \in E$, one defines $l(u,v) = f(u \odot v)$, with $u \odot v \in \F_2^d$ denoting the \emph{fusion} operator of $u$ and $v$ as defined in~\cite{sutner91}. Intuitively, $u \odot v$ is the vector of $d$ variables defined by adding the rightmost coordinate of $v$ to $u$. The output of a CA equipped with rule $f$ corresponds to a path on the edges of $G_f$. As an example, Figure~\ref{fig:example} depicts the de Bruijn graph of rule $f: \F_2^3 \to \F_2$ defined as $f(x_1, x_2, x_3) = x_1 \oplus x_2 \oplus x_3$, whose Wolfram code is 150, along with its truth table.

\begin{figure}[t]
	\begin{minipage}{0.4\textwidth}
		\centering
		\resizebox{!}{5cm}{%
			\Large
			\begin{tikzpicture}
				[->,auto,node distance=1.5cm, every loop/.style={min distance=12mm},
				empt node/.style={font=\sffamily,inner sep=0pt,outer sep=0pt},
				circ node/.style={circle,thick,draw,font=\sffamily\bfseries,minimum
					width=0.8cm, inner sep=0pt, outer sep=0pt}]
				
				\node [empt node] (e1) {};
				\node [circ node] (n00) [above=1.75cm of e1] {$00$};
				\node [circ node] (n01) [right=1.75cm of e1] {$01$};
				\node [circ node] (n10) [left=1.75cm of e1] {$10$};
				\node [circ node] (n11) [below=1.75cm of e1] {$11$};       
				
				\draw [->, thick, shorten >=0pt,shorten <=0pt,>=stealth] (n00) 
				edge[bend left=20] node (f5) [above right]{$1$} (n01);
				\draw [->, thick, shorten >=0pt,shorten <=0pt,>=stealth] (n01)
				edge[bend left=20] node (f5) [below right]{$0$} (n11);
				\draw [->, thick, shorten >=0pt,shorten <=0pt,>=stealth] (n11)
				edge[bend left=20] node (f5) [below left]{$0$} (n10);
				\draw [->, thick, shorten >=0pt,shorten <=0pt,>=stealth] (n10)
				edge[bend left=20] node (f5) [above left]{$1$} (n00);
				\draw[->, thick, shorten >=0pt,shorten <=0pt,>=stealth] (n10) edge[bend
				left=20] node (f1) [above]{$0$} (n01);
				\draw[->, thick, shorten >=0pt,shorten <=0pt,>=stealth] (n01)
				edge[bend left=20] node (f2) [below]{$1$} (n10);
				\draw[->, thick, shorten >=0pt,shorten <=0pt,>=stealth] (n00) edge[loop
				above] node (f3) [above]{$0$} ();
				\draw[->, thick, shorten >=0pt,shorten <=0pt,>=stealth] (n11) edge[loop
				below] node (f4) [below]{$1$} ();
			\end{tikzpicture}
		}
	\end{minipage}%
	\begin{minipage}{0.6\textwidth}
		\centering
		\begin{tabular}{cc|c}
			\hline
			$u \to v$ & $x = u \odot v$ & $f(x) = l_f(u,v)$ \\ 
			\hline
			$00 \to 00$            & 000            &   0      \\
			$10 \to 00$            & 100            &   1      \\
			$01 \to 10$            & 010            &   1      \\
			$11 \to 10$            & 110            &   0      \\
			$00 \to 01$            & 001            &   1      \\
			$10 \to 01$            & 101            &   0      \\
			$01 \to 11$            & 011            &   0      \\
			$11 \to 11$            & 111            &   1      \\
			
			\hline
		\end{tabular}
	\end{minipage}
	\caption{Example of orthogonal labelings for the de Bruijn graph $G_{2,2}$ induced by the CA local rules 90 and 150 of diameter $d=3$.}
	\label{fig:example}
\end{figure}
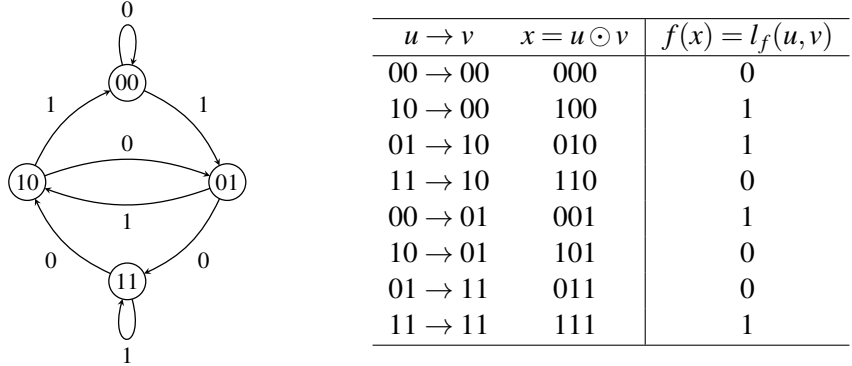

A local rule $f: \Sigma^d \to \Sigma$ is called \emph{leftmost} (respectively, \emph{rightmost}) \emph{permutive} if, by fixing all input coordinates $x_2, \ldots, x_d$ (respectively, $x_1, \ldots, x_{d-1}$) to any constant value, the restriction of $f$ on the remaining variable is a permutation over $\Sigma$. Accordingly, a local rule is called \emph{bipermutive} if it is both leftmost and rightmost permutive. For the binary case, this means that a bipermutive rule $f: \F_2^d \to \F_2$ can be written as $f(x_1,\ldots, x_d) = x_1 \oplus g(x_2, \ldots, x_{d-1}) \oplus x_d$, where $g: \F_2^{d-2} \to \F_2$ is the \emph{generating function} of $f$~\cite{leporati14}). Bipermutive CA are surjective, since the labels of the outgoing (respectively, ingoing) edges of any vertex $v \in V$ of the de Bruijn graph induce a permutation of the alphabet (therefore, one the edges can be traversed in both directions).

Beside binary CA, we consider the case of CA defined over the finite field alphabet $\F_q$, with its size $q$ being a power of a prime number. A  local rule $f: \F_q^d \to \F_q$ is called \emph{linear} if it is defined as a linear combination of the input variables, that is	$f(x_1,\ldots,x_d) = a_1x_1 + \ldots + a_dx_d$ for all $x = (x_1,\ldots,x_d) \in \F_q^d$, where $a_1,\ldots a_d \in \F_q$ and the sum and multiplication correspond to the field operations. A NBCA $F: \F_q^n \to \F_q^{n-d+1}$ defined by a linear local rule can then be expressed as a linear map $F(x) = M_F\cdot x^\top$, where $M_F$ is a $(n-d+1)\times n$ \emph{transition matrix} defined by successive shifts of the coefficients $a_1, \ldots, a_d$. One can also associate to a linear local rule $f$ the polynomial $p_f(X) = a_1 + a_2X + \ldots + a_{d}X^{d-1}$. Under the lens of error-correcting codes, $M_F$ and $p_f$ are respectively the generator matrix and the generator polynomial of a cyclic code~\cite{mariot18}. Clearly, any linear rule is bipermutive if and only if both $a_1$ and $a_d$ are not null. In this case, the CA is called a Linear Bipermutive CA (LBCA).

\subsection{Secret Sharing Schemes}
Informally, a \emph{secret sharing scheme} is a protocol where a \emph{dealer} wishes to share a \emph{secret} $S$ (such as a private key to sign documents) with a set of $n$ participants or players $\mathcal{P} = \{P_1,P_2,\cdots,P_n\}$. During the setup phase of the scheme, the dealer computes $n$ \emph{shares} of $S$, denoted as $B_1,\ldots, B_n$, and distribute them respectively to the players $P_1,\ldots P_n$. At a later stage, any \emph{authorized subset} $A \subseteq \mathcal{P}$ can recover the secret $S$ by pooling together the shares of the corresponding players. The \emph{access structure} $\Gamma \subseteq 2^{\mathcal{P}}$ lists all authorized subsets of the scheme, and it is defined as the union-closure of its \emph{basis} $\Gamma_0$, which specifies all minimal authorized subsets. In $(t,n)$ \emph{threshold schemes} (such as those introduced by Shamir~\cite{shamir79} and Blakley~\cite{blakley79}), every subset of at least $t$ players is authorized; therefore, the basis of the access structure is defined as $\Gamma_0 = \{A \subseteq \mathcal{P}: |A| = t\}$.

The security of a secret sharing scheme is characterized by the information that unauthorized subsets can gain on the value of the secret. In particular, a scheme is called \emph{perfect} if it does not leak any information about $S$ by knowing the shares of any unauthorized subset. More formally, if the secret $S$ is drawn from $\mathcal{S}$ according to a probability distribution $Pr(S)$, then for any unauthorized subset $U \notin \Gamma$ it holds $Pr(S|U) = Pr(S)$. Assuming that a suitable notion of size on the secret and the shares (e.g., the number of bits required to encode them), a perfect secret sharing scheme is called \emph{ideal} if for all $S \in \mathcal{S}$ the size of the shares induced by any secret $S$ equals the size of $S$. 

Most of the constructions for threshold schemes assumes that the identities of the players in an authorized subset are known. This usually requires a separate identification protocol during the reconstruction phase. \emph{Anonymous secret sharing schemes}, initially introduced by Stinson and Vanstone~\cite{stinson88}, circumvent this problem by recovering the secret only as a function of the shares value, without associating them to the players' identities. Phillips and Phillips~\cite{phillips92} analyzed the conditions under which an anonymous secret sharing scheme can be ideal, showing that this property is met only in the $(1,n)$ and $(n,n)$ access structures. Later, Blundo and Stinson~\cite{blundo97} proved lower bounds on the size of the share sets in any $(t,n)$ anonymous secret sharing scheme.

More recently, other works~\cite{eldridge24,bishop25,con25} studied anonymous secret sharing with additional properties, such as unlinkability of shares from different unauthorized subsets. We do not consider such schemes in the scope of this work, referring the reader to~\cite{con25} for more information about them.

\subsection{Latin Squares}
\label{subsec:ls}
We recall only the basic definitions about Latin squares that are relevant for the characterization of $(2,n)$ threshold schemes. A thorough treatment of this subject can be found in~\cite{denes-ls}. Given $n \in \N$, we denote by $[N] = \{1,\ldots, N\}$ the set of the first $N$ integer positive numbers. A Latin square of order $N \in \N$ is an $N \times N$ matrix $L$ with entries from $[N]$ such that for all $i, j, k \in [N]$ with $j\neq k$, one has that $L(i, j) \neq L(i, k)$ and $L(j, i) \neq L(k, i)$. Equivalently, every row and column of a Latin square is a permutation of $[N]$. Two Latin squares $L_1, L_2$ of order $N \in \N$ \emph{orthogonal} if for all distinct pairs of coordinates $(i, j), (i', j') \in [N] \times [N]$ it holds that $(L_1(i,j), L_2(i,j)) \neq (L_1(i', j'), L_2(i', j'))$. In other words, the \emph{superposition} of two orthogonal Latin squares is a permutation of the Cartesian product $[N] \times [N]$. A family of $n$ \emph{Mutually Orthogonal Latin Squares} (also denoted as $n$-MOLS) is a set of $n$ Latin squares that are pairwise orthogonal

As shown for example in~\cite{stinson-cd}, families of $n$-MOLS can be used to construct $(2,n)$ threshold secret sharing schemes. In particular, suppose that $L_1,\ldots L_n$ are Latin squares of order $N$ such that $L_i$ and $L_j$ are orthogonal for each $i\neq j$. We assume that this family is public, hence known to any participant in the scheme (and possibly also to adversaries). Then, one can define the secret and the share spaces as $\mathcal{S} = \mathcal{B} = [N]$. The dealer encodes $S \in [N]$ as a row of the squares, and selects a random column $R \in [N]$. The shares to be distributed to the players are then the entries $L_1(S,R), \ldots , L_n(S,R)$. Later, if any two players $i, j$ combine their shares, then since $L_i$ and $L_j$ are orthogonal they know that the pair $(L_i(S,R), L_j(S,R))$ occurs exactly once the superposed squares, allowing to retrieve uniquely the row and column coordinates, and thus the secret row $S$. On the other hand, if a single player $i$ tries to reconstruct $S$ on their own, the value $L_i(S,R)$ occurs $N$ times on the Latin square $L_i$. Therefore, if the dealer selected the random column $R$ with uniform probability, every value of the secret is equally likely. Thus, the scheme is perfect and also ideal, since the secret and share sets coincide. On the other hand, the scheme is not anonymous, since each player must know their associated Latin square to allow reconstruction.

Mariot et al.~\cite{mariot20} investigated families of MOLS generated by CA, motivated by the application to $(2,n)$ threshold schemes. The first step of is to define the \emph{Cayley table} induced by a NBCA. To this end, assume that a total order is defined on $\Sigma^{d-1}$ and that $\phi: \Sigma^{d-1} \rightarrow [N]$ is a monotone bijective mapping between $\Sigma^{d-1}$ and $[N]$, with the latter having the usual order of natural numbers. The inverse mapping of $\phi$ is denoted by $\psi = \phi^{-1}$. Then, \emph{Cayley table} associated to the NBCA $F: \Sigma^{n} \to \Sigma^{n-d+1}$ is the $N \times N$ matrix $C_F$ with entries from $[N]$ such that $C_{F}(i,j) = \phi(F(\psi(i)||\psi(j)))$ for all $1 \le i,j \le N$, where $||$ denotes concatenation. Thus, the output of the CA corresponds to the entry of the Cayley table at the row and column coordinates represented by the left and the right half of the CA input configuration, respectively.

Next, the authors of~\cite{mariot20} showed that the Cayley table associated to any NBCA $F: \Sigma^{2(d-1)} \rightarrow \Sigma^{d-1}$ equipped with a bipermutive local rule $f:\Sigma^{d}\rightarrow \Sigma$ is a Latin square of order $N=q^{d-1}$\footnote{The same result was independently observed much earlier under a different guise by Eloranta~\cite{eloranta93}.} Then, they proved the following characterization of LBCA pairs whose associated Latin squares are orthogonal:
\begin{theorem}[\cite{mariot20}]
\label{thm:lin-oca}
    Let $F,G:\F_q^{2(d-1)} \to \F_q^{d-1}$ be two linear bipermutive CA, and let $p_f,p_g \in \F_q[X]$ be the respective associated polynomials. Then, the corresponding Latin squares $C_F$ and $C_G$ of order $N=q^{d-1}$ are orthogonal if and only if $\gcd(p_f,p_g) = 1$, i.e. if and only if $p_f$ and $p_g$ are relatively prime.
\end{theorem}
Hence, one can construct a family of $n$-MOLS (or equivalently, a $(2,n)$ scheme) based on LBCA by finding a set of $n$ polynomials of degree $k=d-1$ and with a nonzero constant term that are pairwise coprime. The authors of~\cite{mariot20} further proved that the size of a maximal family of Mutually Orthogonal Latin Squares based on linear CA (also called linear MOCA) equals:
\begin{equation}
	\label{eq:gauss}
	N_k = I_k + \sum_{j=1}^{\lfloor n/2 \rfloor} I_j \enspace ,
\end{equation}
where $I_r$ denotes for $r \in \N$ the number of irreducible polynomials of degree $r$ over $\F_q$, computed using Gauss's formula~\cite{gauss-irr}.

\section{Anonymous Secret Sharing with CA}
\label{sec:scheme}

\subsection{Basic $(2, n)$ Threshold Scheme based on MOCA}
We start by first recalling how the basic $(2, n)$ threshold scheme based on MOCA proposed in~\cite{mariot18} works. Let $\Sigma$ be an alphabet of $q$ symbols, $d \in \N$ and assume that we have $n$ players $P_1, P_2, \ldots, P_n$. The secret and share spaces coincide with the set $\Sigma^{d-1}$ of configurations of $d-1$ cells. Hence, we have $|\mathcal{S}| = |\mathcal{B}| = q^{d-1}$. The dealer publishes a set of $n$ bipermutive local rules $f_1, \ldots f_n: \Sigma^{d} \to \Sigma$, not necessarily linear, that give rise to a set of $n$-MOCA.

\paragraph{Setup Phase.} The setup phase of the protocol works as follows:
\begin{enumerate}
	\item The dealer selects the secret $S \in \Sigma^{d-1}$ to be shared, and concatenates it with a random block $R \in \Sigma^{d-1}$, thus obtaining the input configuration $x = S||R \in \Sigma^{2(d-1)}$.
	\item For each $i \in [n]$, the dealer computes the NBCA $F_i: \Sigma^{2(d-1)} \to \Sigma^{d-1}$ equipped with rule $f_i$ on the configuration $x$, obtaining $B_i = F_i(x) \in \Sigma^{d-1}$
	\item For each $i \in [n]$, the dealer sends $B_i$ to the player $P_i$.
\end{enumerate}

\paragraph{Recovery Phase.} Given two players $P_i, P_j$ and their respective shares $B_i, B_j$, the recovery phase unfolds in the following steps:
\begin{enumerate}
	\item $P_i$ and $P_j$ fetch the CA local rules $f_i, f_j$ published by the dealer, and construct the corresponding \emph{coupled de Bruijn graph} $G_{i,j}$, which is simply the superposition of the de Bruijn graphs of $f_i$ and $f_j$. Thus, each edge $(u,v)$ of $G_{i,j}$ is labeled with with the ordered pair $(f_i(u \odot v), f_j(u \odot v))$.
	\item $P_i$ and $P_j$ pool their shares $B_i, B_j$ and find the path $W$ on the edges of $G_{i,j}$ that is labelled as $((B_{i,1}, B_{j,1}), \ldots, (B_{i,d-1}, B_{j,d-1}))$. Notice that there is a unique such path, since the Latin squares of $f_i$ and $f_j$ are orthogonal.
	\item Determine the original input configuration $x \in \Sigma^{2(d-1)}$ by computing the fusion operator $\odot$ on the vertices visited by the path $W$. The secret can then be returned as the left half of $x$.
\end{enumerate}
In particular, the reconstruction of the unique path on the edges labelled by the two superposed shares can be performed with the {\sc Invert-OCA} algorithm described in~\cite{mariot18}. Although this scheme is perfect and ideal, as in any MOLS-based construction, it does not allow for anonymous reconstruction, since the players $i$ and $j$ need to know the corresponding local rules $f_i$ and $f_j$ that the dealer used to compute their shares.

\subsection{Modified MOCA-based Anonymous Scheme}
The scheme described above encodes the secret and the shares as blocks of $d-1$ cells, respectively encoding the row and the entries of the Latin squares generated by the family of MOCA. Each MOCA, on the other hand, is assigned to a player.

We now take a different approach: \emph{we encode the secret as one of the Latin squares, or equivalently one of the CA rules}, i.e. $\mathcal{S} = \{f_1, \ldots, f_n\}$. On the other hand, the space of the shares becomes $\mathcal{B} = \Sigma^{2(d-1)}$, that is, the set of all possible input configurations of the CA. Further, the number of players that can participate in the scheme now corresponds to the number of preimages that map to a specific output block $y \in \Sigma^{d-1}$, which are $q^{d-1}$ due to the balancedness property of surjective CA~\cite{mariot17}.

\paragraph{Setup Phase.} The setup phase of the modified scheme is as follows:
\begin{enumerate}
	\item The dealer chooses the secret CA rule index $S \in \{1, \ldots, n\}$ and selects a random output block of $d-1$ cells $R \in \Sigma^{d-1}$
	\item The dealer computes the counterimage of $S$ under the secret rule $f_S$, that is the set $F_S^{-1}(R) = \{x \in \Sigma^{2(d-1)}: F_S(x) = R\}$. This can be done by constructing the $q^{d-1}$ paths on the edges of the de Bruijn graph $G_{f_S}$ labelled with $R$, starting from each vertex $v \in \Sigma^{d-1}$.
	\item Denoting $F_S^{-1}(R) = \{B_1, B_2, \ldots, B_{q^{d-1}}\}$, the dealer distributes the share $B_i$ to the player $P_i$, for all $i \in [q^{d-1}]$.
\end{enumerate}
Thus, the setup phase entails selecting a \emph{single} CA in the MOCA family as the secret, and then computing the set of preimages of a random output block of this CA, which constitute the shares distributed to the players.

\paragraph{Recovery Phase.} For the recovery, we assume that the dealer published all local rules $f_1, \ldots, f_n$. Suppose now that players $P_i$ and $P_j$ wish to reconstruct the secret. The recovery starts with a non-interactive step, where both players perform computations on their own without interacting with each other. We describe below the non-interactive step of $P_i$ (the step of $P_j$ is completely symmetrical, changing each $i$ index to $j$):
\begin{enumerate}
	\item For all $k \in [n]$, player $P_i$ evaluates the $k$-th CA over their share, i.e. $P_i$ computes $y_{i, k} = F_k(B_i)$.
	\item For all $k \in [n]$, $P_i$ computes the counterimage of $y_{i,k}$ under rule $f_k$, i.e. the set $A_{i,k} = F_k^{-1}(y_{i,k}) = \{x \in \Sigma^{2(d-1)}: F_k(x) = y_{i,k}\}$.
\end{enumerate}
At this point, $P_i$ and $P_j$ have constructed two families of counterimages, respectively $\mathcal{A}_{i} = \{A_{i,1}, \ldots, A_{i,n}\}$ and $\mathcal{A}_{j} = \{A_{j,1}, \ldots, A_{j,n}\}$. The interactive phase now works as follows:
\begin{enumerate}
	\item $P_i$ and $P_j$ compute the intersection $\mathcal{A}_i \cap \mathcal{A}_j$.
	\item The result of the intersection uniquely identifies the secret CA rule $f_S$ selected by the dealer during the setup phase.
\end{enumerate}

\subsection{Analysis of the Scheme}
In what follows, we prove the correctness of the above scheme and analyze its security, showing that it is both perfect and anonymous.

\begin{lemma}
	The intersection computed by $P_i$ and $P_j$ during the last step of the protocol yields a single element, and this element is the counterimage $A_S = F_S^{-1}(R)$ computed by the dealer in the setup phase, which uniquely identifies the secret CA rule $f_S$.
\end{lemma}
\begin{proof} For any $i \in [n]$, consider the CA $F_i: \Sigma^{2(d-1)} \to \Sigma^{d-1}$ equipped with rule $f_i: \Sigma^{d}  \to \Sigma^d$. The counterimages $F_i^{-1}(y)$ of the output blocks $y \in \Sigma^{d-1}$ of the CA partition the input space of $\Sigma^{2(d-1)}$. In fact, fixing an output block $y \in \Sigma^{d-1}$ and determining its counterimage basically means to construct the set of coordinates that feature $y$ as an entry on the Latin square of $F_i$. Due to the property of the Latin square, two different outputs $y\neq y'$ cannot share the same row and column coordinates. In the terminology of combinatorial designs, the partition $\Pi_i= \{F_i^{-1}(y): y \in \Sigma^{d-1}\}$ is also called a $1$-\emph{parallel class}, since it is a set of subsets all of the same size ($q^{d-1}$, due to the balancedness property) where each $x \in \Sigma^{2(d-1)}$ appears exactly once.

Consider now the set of all parallel classes $\Pi_1, \ldots, \Pi_n$ induced by the $n$-MOCA $f_1,\ldots, f_n$. Then, any two subsets belonging to different classes $\Pi_i, \Pi_j$ can intersect in at most one preimage. Indeed, suppose that $\Pi_i$ and $\Pi_j$ have two subsets of counterimages intersecting on more than one preimage: then, it means that the superposition of the Latin squares generated by $f_i, f_j$ have a repeated pair of entries, contradicting the orthogonality property. 

Therefore, since the shares held by $P_i$ and $P_j$ have been taken from the \emph{same} subset of preimages $F_S^{-1}(R)$ generated by the dealer in the setup phase, by construction $\mathcal{A}_i \cap \mathcal{A}_j$ can intersect only in that subset.
\end{proof}

Once the two players have identified the unique parallel class $\Pi_S$ to which the subset resulting from the intersection belongs to, they can easily determine what is the secret CA rule $f_S$. The next result proves that the new scheme is also perfect, and moreover it allows anonymous reconstruction:

\begin{lemma}
	\label{lm:sec}
	The modified MOCA-based scheme is a perfect and anonymous $(2,q^{d-1})$ threshold secret sharing scheme.
\end{lemma}
\begin{proof}
	Without loss of generality, suppose that the player $P_i$ tries to reconstruct the secret CA rule starting from the share $B_i$. Then, what $P_i$ can do is to follow the non-interactive step of the recovery phase, computing the output blocks $y_{i,k} = F_k(B_i$ for all CA $f_k$ with $k \in [n]$, and then constructing the corresponding sets of counterimages $A_{i,k} = F_k^{-1}(y_{i,k})$. By construction, each counterimage belongs to a distinct parallel class $\Pi_k$, and they all intersect on $B_i$. Thus, under the assumption that the dealer chose the random output block $R \in \Sigma^{d-1}$ with uniform probability during the setup phase, any of the counterimages generated by $P_i$ is equally likely to be the one computed by the dealer. Therefore, $P_i$ does not obtain any information on the correct parallel class $\Pi_S$, and thus on the secret CA rule $f_S$. 	Concerning anonymity, remark that the intersection of $\mathcal{A}_i \cap \mathcal{A}_j$ returns a subset of preimages, which contains the shares $B_i$ and $B_j$. However, the reconstruction phase does not require to identify $B_i$ and $B_j$ as the shares held by $P_i$ and $P_j$, since the subset is directly associated to the parallel class $\Pi_S$. Thus, the scheme also allows for anonymous reconstruction.
\end{proof}

\begin{remark}
There is an important caveat to consider in the anonymity proof of Lemma~\ref{lm:sec}: namely, how to compute the intersection of $\mathcal{A}_i$ and $\mathcal{A}_j$, which are held respectively by $P_i$ and $P_j$. The most straightforward solution is that both $P_i$ and $P_j$ disclose to each other their sets $\mathcal{A}_i$ and $\mathcal{A}_j$, and then compute the intersection together. However, this procedure would break anonymity. For example, suppose that $P_i$ wants to find out the share held by $P_j$. Then, $P_i$ can simply check what is the value that occur in every preimage set received from $P_j$. By construction, such value corresponds to $B_j$. To solve this issue, the two players can resort to a \emph{Private Set Intersection} (PSI) protocol, which allows two parties to compute the intersection of their private sets without disclosing them~\cite{freedman04}.
\end{remark}

\begin{example}
We conclude this section by presenting a fully worked out example of the scheme described so far. Let $q=2$ and $d=3$, i.e. we work with binary elementary CA. In this case, there are only two orthogonal CA, namely those defined by the linear local rules $f(x_1,x_2,x_3) = x_1 \oplus x_3$ and $g(x_1,x_2,x_3) = x_1 \oplus x_2 \oplus x_3$, i.e. rule 90 and 150. By Theorem~\ref{thm:lin-oca} these rules form orthogonal Latin squares of order $2^{3-1}=4$, because their associated polynomials $p_f(X) = 1+X^2$ and $p_g(X) = 1+X+X^2$ are coprime. Figure~\ref{fig:r150-90} displays the Latin squares associated to rule 90 and 150, and their partition in parallel classes. Remark that in the parallel classes we encode the row-column coordinates $(i,j) \in [4] \times [4]$ as single integer numbers from $1$ to $16$, using lexicographic ordering.
\begin{figure}[t]
	\centering
		\begin{subfigure}{.3\textwidth}
		\centering
		\begin{tikzpicture}
			[->,auto,node distance=1.5cm,
			empt node/.style={font=\sffamily,inner sep=0pt,minimum size=0pt},
			rect node/.style={rectangle,draw,font=\sffamily,minimum size=0.7cm, inner sep=0pt, outer sep=0pt}]
			
			\node [rect node] (s11) {$1$};
			\node [rect node] (s12) [right=0cm of s11] {$2$};
			\node [rect node] (s13) [right=0cm of s12] {$3$};
			\node [rect node] (s14) [right=0cm of s13] {$4$};
			
			\node [rect node] (s21) [below=0cm of s11] {$2$};
			\node [rect node] (s22) [right=0cm of s21] {$1$};
			\node [rect node] (s23) [right=0cm of s22] {$4$};
			\node [rect node] (s24) [right=0cm of s23] {$3$};
			
			\node [rect node] (s31) [below=0cm of s21] {$3$};
			\node [rect node] (s32) [right=0cm of s31] {$4$};
			\node [rect node] (s33) [right=0cm of s32] {$1$};
			\node [rect node] (s34) [right=0cm of s33] {$2$};
			
			\node [rect node] (s41) [below=0cm of s31] {$4$};
			\node [rect node] (s42) [right=0cm of s41] {$3$};
			\node [rect node] (s43) [right=0cm of s42] {$2$};
			\node [rect node] (s44) [right=0cm of s43] {$1$};
			
		\end{tikzpicture}
		\caption{Rule $90$}
	\end{subfigure}%
	\begin{subfigure}{.3\textwidth}
		\centering
		\begin{tikzpicture}
			[->,auto,node distance=1.5cm,
			empt node/.style={font=\sffamily,inner sep=0pt,minimum size=0pt},
			rect node/.style={rectangle,draw,font=\sffamily,minimum size=0.7cm, inner sep=0pt, outer sep=0pt}]
			
			\node [rect node] (s11) {$1$};
			\node [rect node] (s12) [right=0cm of s11] {$4$};
			\node [rect node] (s13) [right=0cm of s12] {$3$};
			\node [rect node] (s14) [right=0cm of s13] {$2$};
			
			\node [rect node] (s21) [below=0cm of s11] {$2$};
			\node [rect node] (s22) [right=0cm of s21] {$3$};
			\node [rect node] (s23) [right=0cm of s22] {$4$};
			\node [rect node] (s24) [right=0cm of s23] {$1$};
			
			\node [rect node] (s31) [below=0cm of s21] {$4$};
			\node [rect node] (s32) [right=0cm of s31] {$1$};
			\node [rect node] (s33) [right=0cm of s32] {$2$};
			\node [rect node] (s34) [right=0cm of s33] {$3$};
			
			\node [rect node] (s41) [below=0cm of s31] {$3$};
			\node [rect node] (s42) [right=0cm of s41] {$2$};
			\node [rect node] (s43) [right=0cm of s42] {$1$};
			\node [rect node] (s44) [right=0cm of s43] {$4$};
			
		\end{tikzpicture}
		\caption{Rule $150$}
	\end{subfigure}%
	\begin{subfigure}{.4\textwidth}
		\centering
		\setlength{\tabcolsep}{2pt}
		\begin{tabular}{c|cc}
		$y$ & $\Pi_{90}$ & $\Pi_{150}$ \\
		\hline
		00 (1) & $\{1,6,11,16\}$ & $\{1,8,10,15\}$ \\
		10 (2) & $\{2,5,12,15\}$ & $\{4,5,11,14\}$ \\
		01 (3) & $\{3,8,9,14\}$ & $\{3,6,12,13\}$ \\
		11 (4) & $\{4,7,10,13\}$ & $\{2,7,9,16\}$ \\
		\end{tabular}
		
		\phantom{M}
		
		\caption{Parallel classes}
	\end{subfigure}%
	
	\caption{Latin squares generated by rules 90 and 150, and related parallel classes.}
	\label{fig:r150-90}
\end{figure}
The secret space is thus composed only of two rules, $\mathcal{S}={\{f_{90}, f_{150}\}}$. The share set coincides with the set $\F_2^4=\{0,1\}^4$ of all $16$ input configurations of $4$ bits, and the scheme has $2^{3-1}=4$ players, i.e. $\mathcal{P} = \{P_1, P_2, P_3, P_4\}$. Suppose now that the dealer selects $f_{90}$ as the secret rule, and samples $10$ as a random output block configuration. Thus, the shares are calculated by computing the $4$ preimages of $10$ under rule 90. From the table of the parallel classes decomposition, this corresponds to the set $\{2,5,12,15\}$ and the shares therein are distributed to the four players: $B_1 = 2$, $B_2 = 5$, $B_3 = 12$ and $B_4 = 15$. Later, let us assume that $P_1$ and $P_3$ want to reconstruct the secret rule by combining their shares. Then, $P_1$ and $P_3$ evaluate the CA $F_{90}, F_{150}: \F_2^4 \to \F_2^2$ respectively on their shares $B_1$ and $B_3$ as inputs, thus obtaining: $y_{1,90} = F_{90}(2) = 10$, $y_{1,150} = F_{150}(2) = 11$, $y_{3,90} = F_{90}(12) = 10$, and $y_{3,150} = F_{150}(12) = 01$. Next, $P_1$ and $P_3$ construct the preimage sets of the output blocks computed in the previous step, thus obtaining the sets:
\begin{align*}
	\mathcal{A}_1 &= \{\{2,5,12,15\}, \{2,7,9,16\}\} \enspace , \\
	\mathcal{A}_3 &= \{\{2,5,12,15\}, \{3,6,12,13\}\} \enspace .
\end{align*}
Finally, the intersection $\mathcal{A}_1 \cap \mathcal{A}_3$ yields the subset $\{2,5,12,15\}$, which allows to uniquely identify the parallel class $\Pi_{90}$, thus recovering the secret rule $f_{90}$.
\end{example}

\section{Discussion and Open Problems}
\label{sec:outro}
While the scheme presented in the previous section is conceptually simple from an abstract point of view, some considerations on its computational complexity are in order. Both in the setup and the recovery phase, the most expensive operation is computing the set of preimages of an output block under the action of a bipermutive CA. In particular, the number of preimages is exponential in the size $d-1$ of an output block, which is determined by the diameter of the CA. Moreover, while in the setup phase  the dealer only constructs a single preimage set (the one corresponding to the random output block $R$ under the action of the secret CA), recovery requires that both players construct a set of preimages \emph{for each CA rule}. This adds a multiplicative factor of $n$, giving an overall complexity of $\mathcal{O}(nq^{d})$\footnote{Of course, we are factoring out the complexity required to compute a single preimage.} for a player $P_i$ to compute their set $\mathcal{A}_i$. Consequently, the higher is the number of CA rules that can be possibly shared, the larger will be the time required for the players to reconstruct the secret. 

These remarks prompt us with two interesting open problems for future research. The first one is to find a more efficient algorithm to determine the common parallel class between two players, that does not require to compute the full preimage sets. An interesting approach here could be to explore secure multiparty computation techniques over the de Bruijn graph representation of the local rules, to check on the fly whether the preimage sets that is being computed by the two players is the same or not. The second problem is to investigate more in depth the \emph{information ratio} of this scheme, i.e. the ratio between the size of the shares and the size of the secrets. Clearly, the scheme described in this paper is not ideal, as the number of shares (which correspond to blocks of size $d-1$) is always higher than the number of secrets that can be shared, which correspond to the maximum number of MOCA that can be constructed for a given diameter. We also note that, while the construction for maximal MOCA families exhibited in~\cite{mariot20} based on linear CA is the best known so far, it is not known whether one can construct even larger families of MOCA with nonlinear CA.

\bibliographystyle{abbrv}
\bibliography{bibliography}

\end{document}